\documentclass[11pt,english]{article}

\usepackage[utf8]{inputenc}
\usepackage[T1]{fontenc}
\usepackage{lmodern}
\usepackage{microtype}

\usepackage{amssymb}
\usepackage{color}
\usepackage{amsmath}
\usepackage{amssymb}
\usepackage{mathtools}
\usepackage{mathrsfs}

\usepackage{booktabs}
\usepackage{array}
\newcolumntype{L}[1]{>{\raggedright\arraybackslash}m{#1}}
\newcolumntype{C}[1]{>{\centering\let\newline\\\arraybackslash\hspace{0pt}}m{#1}}
\newcolumntype{R}[1]{>{\raggedleft\arraybackslash}m{#1}}
\usepackage{threeparttable}

\usepackage{cite}
\usepackage{url}
\usepackage[svgnames]{xcolor}
\usepackage{xspace}
\usepackage{needspace}
\usepackage[inline]{enumitem}

\usepackage{pifont}

\usepackage{tikz}
\usetikzlibrary{calc,fadings,shapes.arrows,shadows,backgrounds,positioning}

\usepackage{float}
\usepackage{graphicx}
\usepackage{caption,subcaption}

\usepackage{amsthm}

\usepackage{algorithm}
\usepackage[noend]{algpseudocode}

\usepackage[capitalize]{cleveref}

\theoremstyle{plain}
\newtheorem{theorem}{Theorem}
\newtheorem{lemma}[theorem]{Lemma}
\newtheorem{proposition}[theorem]{Proposition}

\newtheorem{fact}[theorem]{Fact}
\newtheorem*{remark}{Remark}
\theoremstyle{definition}
\newtheorem{definition}[theorem]{Definition}

\newsavebox{\fboxenvbox}

\newcommand{\Z}{\mathbb{Z}}
\newcommand{\N}{\mathbb{N}}

\makeatletter
\newcommand\suchthat{%
 \@ifstar
  {\mathrel{}\middle|\mathrel{}}
  {\mid}%
}
\makeatother

\newcommand{\ve}{\varepsilon}
\newcommand{\eps}{\ve}

\let\leq\leqslant 
\let\geq\geqslant 

\newcommand{\level}{{\ell}}

\newcommand{\levelmax}{L}
\newcommand*{\true}{\textsf{true}\xspace}

\newcommand{\OPT}{\textsc{OPT}}

\newcommand{\csP}{\textsf{\#P}\xspace}

\newcommand{\AND}{\ensuremath{\mathsf{and}}\xspace}
\newcommand{\XOR}{\ensuremath{\mathsf{xor}}\xspace}
\newcommand{\NOT}{\ensuremath{\mathsf{not}}\xspace}

\newcommand{\len}{\mathsf{len}}

\usepackage{bbm}

\usepackage{fullpage}
\usepackage[textwidth=20mm]{todonotes}
\usepackage{xspace}

\ifdefined\DEBUG

 \newcommand{\fab}[1]{\textcolor{blue}{#1}}
 \newcommand{\ta}[1]{\textcolor{cyan}{#1}}
 \newcommand{\cla}[1]{\textcolor{green}{#1}}
 \newcommand{\han}[1]{\textcolor{red}{#1}}
  
  \newcommand{\mytodo}[2]{\todo[size=\scriptsize, color=#1!50!white]{#2}}
  
  \newcommand{\fabr}[1]{\mytodo{blue}{#1}}
  \newcommand{\tar}[1]{\mytodo{cyan}{#1}}
  \newcommand{\clar}[1]{\mytodo{green}{#1}}
  \newcommand{\hanr}[1]{\mytodo{red}{#1}}
  
\else

  \newcommand{\fab}[1]{#1}
  \newcommand{\ta}[1]{#1}
  \newcommand{\cla}[1]{#1}
  \newcommand{\han}[1]{#1}

  \newcommand{\fabr}[1]{}
  \newcommand{\tar}[1]{}
  \newcommand{\clar}[1]{}
  \newcommand{\hanr}[1]{}

\fi

\title{Optimization of Bootstrapping in Circuits}

\author{
\fab{Fabrice Benhamouda}\thanks{ENS, CNRS, INRIA, and PSL Research University, Paris, France. \texttt{fabrice.benhamouda@ens.fr}. Research supported in part by the CFM Foundation and the French FUI Project FUI AAP 17 CRYPTOCOMP.}
\and
\ta{Tancrède Lepoint}\thanks{SRI International, USA. \texttt{tancrede.lepoint@sri.com}.}
\and
\cla{Claire Mathieu}\thanks{ENS, CNRS, and PSL Research University, Paris, France. \texttt{cmathieu@di.ens.fr}.}
\and
\han{Hang Zhou}\thanks{Max-Planck-Institut f\"{u}r Informatik, Saarbrücken, Germany. \texttt{hzhou@mpi-inf.mpg.de}. Research supported in part by the Lise Meitner Award Fellowship.}
}
\date{}

\begin{document}

\maketitle

\begin{abstract}
In 2009, Gentry proposed the first Fully Homomorphic Encryption (FHE) scheme, an extremely powerful cryptographic primitive that enables to perform computations, i.e., to evaluate circuits, on encrypted data without decrypting them first. This has many applications, in particular in cloud computing.

In all currently known FHE schemes, encryptions are associated to some (non-negative integer) noise level, and at each evaluation of an AND gate, the noise level increases.
This is problematic because decryption can only work if the noise level stays below some maximum level~$L$ at every gate of the circuit.
To ensure that property, it is possible to perform an operation called \emph{bootstrapping} to reduce the noise level. However, bootstrapping is time-consuming and has been identified as a critical operation. This motivates a new problem in discrete optimization, that of choosing where in the circuit to perform bootstrapping operations so as to control the noise level; the goal is to minimize the number of bootstrappings in circuits.

In this paper, we formally define the \emph{bootstrap problem}, we design a polynomial-time \mbox{$L$-approximation} algorithm using a novel method of rounding of a linear program, and we show a matching hardness result: $(L-\eps)$-inapproximability for any $\eps>0$.

\addvspace{\baselineskip}
\noindent
\textbf{Keywords.} DAG, circuit, approximation algorithms, inapproximability, linear programming, rounding, fully homomorphic encryption (FHE)
\end{abstract}

\pagenumbering{roman}
\thispagestyle{empty}

\newpage{}
\setcounter{page}{1}
\pagenumbering{arabic}

\section{Introduction}
\label{sec:intro}

Imagine evaluating a circuit with noise: at each gate the noise level may increase due to the computation. Now, imagine that you can occasionally perform a computationally expensive operation on the output of a gate (called \emph{bootstrapping}) to reduce the noise level. Given a circuit, at which gates should you apply the bootstrapping operation to the output of the gate so that the maximum noise level remains within a certain tolerance level?  We want to minimize the number of bootstrappings. 

For example, if the noise level at an input gate equals 0 and the noise level at a gate with two direct predecessors $u$ and $v$ equals $\max(\text{noiselevel}(u),\text{noiselevel}(v))+1$,  then in the circuit in~\cref{fig:intro-example}, it is possible to maintain a maximum noise level of at most $L=3$ by doing 2 bootstrappings; that is optimal.

\begin{figure}[hb]
\centering
\begin{subfigure}[t]{0.325\textwidth}\centering
  \includegraphics[scale=1.1]{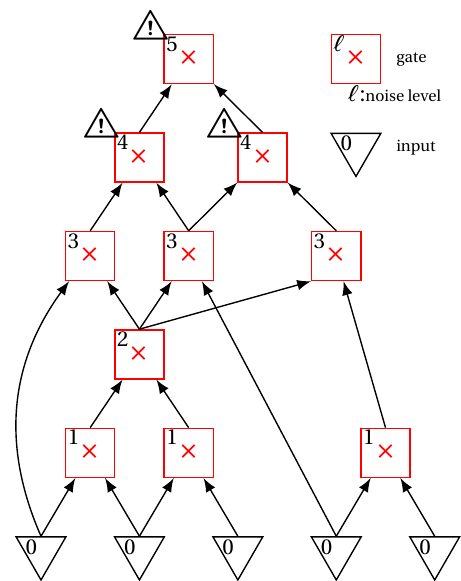}
  \caption{Original circuit:\\ \footnotesize without bootstrapping, the noise of the output of the three top gates is above $L=3$ (warning sign).}
\end{subfigure}~
\begin{subfigure}[t]{0.325\textwidth}\centering
  \includegraphics[scale=1.1]{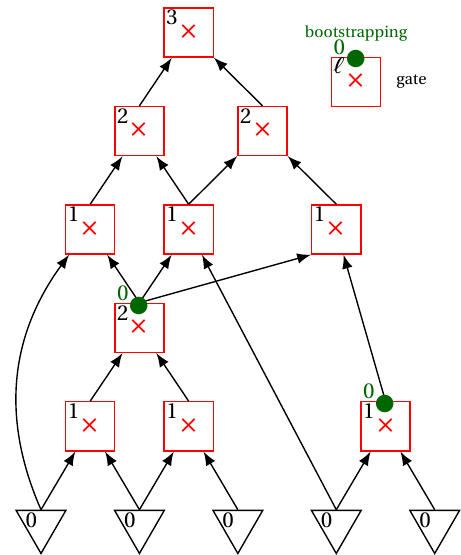}
  \caption{Optimal bootstrapping strategy:\\ \footnotesize the noise of the output of all gates is below $L=3$.}
\end{subfigure}~
\begin{subfigure}[t]{0.325\textwidth}\centering
  \includegraphics[scale=1.1]{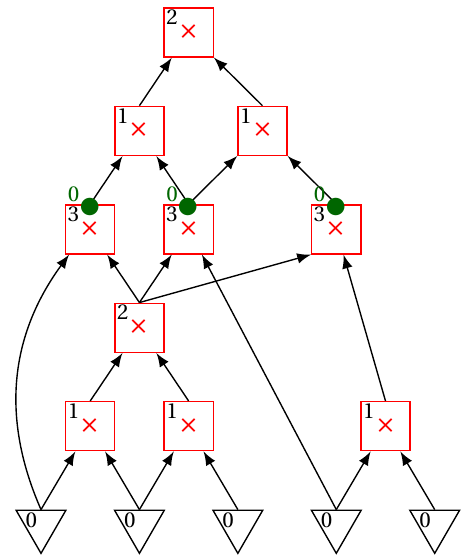}
  \caption{Naive boostrapping strategy:\\ \footnotesize the noise of the output of all gates is below $L=3$ but the strategy requires 3 bootstrappings (non-optimal).}
\end{subfigure}
{\footnotesize }
\caption{Example of circuit and of two bootstrapping strategies}
\label{fig:intro-example}
\end{figure}

\subsection{Motivation}
\label{sec:motivation}

This problem arises in cryptography in the context of \emph{fully homomorphic encryption (FHE)}~\cite{STOC:Gentry09,C:Gentry10,EC:DGHV10,C:CMNT11,EC:CorNacTib12,EC:CCKLLT13,FOCS:GenHal11,C:GenHalSma12,STOC:LopTroVai12,C:GenSahWat13,C:AlpPei13,TOCT:BraGenVai14,SIAM:BraVai14,ITCS:BraVai14,EC:DucMic15}. A fully homomorphic encryption scheme enables one to encrypt bits and keep them confidential, while allowing  {anyone} who is given an encryption $E(a)$ of a bit $a$ and an encryption $E(b)$ of a bit $b$ to publicly compute $E(\NOT\,a)$, $E(a\,\XOR\,b)$, and $E(a\,\AND\,b)$. Such a scheme makes it possible to securely compute {any} binary
circuit over encrypted bits. 
This primitive has tremendous potential for applications, the canonical one being to the problem of outsourcing
computation to a remote server without compromising one's privacy.
Concrete application examples include biometric identification, statistics over encrypted
data~\cite{DBLP:conf/ccs/NaehrigLV11,DBLP:journals/jbi/BosLN14}, machine
learning~\cite{ICISC:GraLauNae12}, and private genomic
analyses~\cite{LC:LauLopNae14}.

All existing instantiations of FHE follow the same blueprint~\cite{STOC:Gentry09}: ciphertexts (i.e., encryptions of bits) contain
some ``noise'' that grows during the circuit evaluation. To ensure correctness
at decryption time, one has to regularly perform \emph{bootstrapping}
operations on the ciphertexts whose aim is to lessen the noise.\footnote{An upper bound on the
admissible noise is given as part of the parameters of the FHE scheme.} Unfortunately, such operations are very expensive in practice (see, e.g.,~\cite{EC:GenHal11,EC:CCKLLT13,FCW:RohCou14,EC:HalSho15,EC:DucMic15,EC:NuiKur15}), hence the  question:
\begin{center}
  \emph{
    What is a minimum set of ciphertexts to be bootstrapped in order to correctly evaluate the circuit?
  }
\end{center} 
This is called the \emph{bootstrap problem}.


\bigskip
In all efficient implementations of FHE schemes~\cite{C:GenHalSma12,STOC:LopTroVai12,FCW:RohCou14,TOCT:BraGenVai14,SIAM:BraVai14,C:HalSho14,EC:HalSho15,EC:DucMic15}, non-linear gates
(\AND) introduce much more additional noise than linear gates  (\NOT and \XOR), hence a simplified model where evaluation of linear
gates do not increase the noise; see,
e.g.,~\cite{FCW:LepPai13,C:HalSho14,EC:ARSTZ15,SAC:PaiVia15,fse:CanteautCFLNPS16}.

Formally, each ciphertext has an associated ``noise level'' $\level\in 
\Z_{\ge 0}$; 
evaluating a \NOT gate over a ciphertext does not change its noise
level; evaluating an \XOR gate yields a ciphertext whose noise level $\max(\cdot, \cdot)$ is the
maximum noise level of its inputs;
on the other hand, evaluating a non-linear gates (\AND) yields a ciphertext with increased noise level $\max(\cdot, \cdot)+1$
(noise behaviors of other FHE schemes are discussed later).

To ensure that the circuit evaluation is correct, the FHE
scheme has a parameter $\levelmax\geq 1$, which is independent of the circuit size, and requires that all ciphertexts must have their noise
levels less than or equal to~$L$ at all gates of the circuit being evaluated. 
This requires performing  a bootstrapping operation on the output of some gates of the circuit; a bootstrapping operation 
reduces the noise level of the ciphertext to~$0$.
The first instantiations of FHE were for $\levelmax=1$~\cite{EC:GenHal11,C:CMNT11,EC:CorNacTib12,EC:CCKLLT13}. Most of
them merely perform a bootstrapping operation right after (see, e.g.,~\cite{EC:GenHal11})
or right before (see, e.g.,~\cite{C:GenHalSma12}) each \AND gate evaluation. However, this can be computationally wasteful since fewer bootstrappings may be sufficient to
evaluate the whole circuit when positioned more carefully~\cite{FCW:LepPai13,SAC:PaiVia15}; see also~\cref{fig:intro-example} (in this figure, square gates correspond to \AND gates). 

Lepoint and Paillier~\cite{FCW:LepPai13} modeled the problem of constructively computing the exact minimum number of
bootstrappings for any $\levelmax\geq 1$ based on Boolean satisfiability. They associated a Boolean to each ciphertext during the circuit evaluation.
The Boolean is equal to \true when the ciphertext should be bootstrapped.
Using the logic circuit and the noise level constraints, the authors constructed a Boolean monotone predicate $\phi$ which captures the
correctness of the circuit evaluation.
They then described a heuristic method
to recover the smallest prime implicant of $\phi$, which directly yields the
minimum number of bootstrappings and the ciphertexts to be bootstrapped.
However no
complexity analysis nor hardness result were claimed in~\cite{FCW:LepPai13}.\footnote{Lepoint and Paillier only indicated that for Boolean monotone predicates, finding the
size of the smallest prime implicant is known to be \csP-complete~\cite[Section~6]{IANDC:GolHagMun08}.} 

Later, Paindavoine and Vialla~\cite{SAC:PaiVia15} showed that, for $L=1$, the bootstrap problem can be solved in polynomial time by a reduction to $(s,t)$-min-cut; and that, for $L\geq2$, the bootstrap problem is NP-hard by a reduction from the vertex cover problem.\footnote{We note that in the terminology of \cite{SAC:PaiVia15}, $l_{max}=L+1$ since the minimum noise level in their model is $1$.} 
They also provided experimental results on real-world circuits (namely integer addition, integer multiplication, and some cryptographic primitives), based on mixed integer linear programming.\footnote{Their linear programming relaxation is different from the one in this paper.}

\subsection{Problem Formulation}
\label{sec:formulation}

In graph theory terms, the bootstrap problem can be formulated as follows: the input is a positive integer $L$ and a directed acyclic graph (DAG) $G=(V,E)$ whose vertices all have indegree~0 or 2, with colors on the vertices: vertices of indegree 0 are white and vertices of indegree 2 are either blue or red. $G$ may have parallel edges. 
A \emph{feasible solution} is a subset $S\subseteq V$ of \emph{marked} vertices such that $\max_{u\in V} \level (u)\leq L$, where the function $\level(\cdot)$ is computed recursively as follows:\footnote{The indicator function $\mathbbm{1}_{V \setminus S}$ has value $1$ on $V \setminus S$ and $0$ on $S$.}
\begin{align*}
  \level(v) &=\begin{cases}
    0 & \text{if $v$ is white,} \\
    \displaystyle \max_{(u,v) \in E} \level(u) \cdot \mathbbm{1}_{V \setminus S}(u) & \text{if $v$ is blue,}\\
    \displaystyle \max_{(u,v) \in E} \level(u) \cdot \mathbbm{1}_{V \setminus S}(u) \;+\; 1& \text{if $v$ is red.}
  \end{cases}
\end{align*}
The goal of the bootstrap problem is to find a feasible solution $S$ of minimum cardinality.

If $S=V$, then $\level(v) \le 1$ for every vertex $v$, so this solution is always feasible. If $S=\emptyset$, then $\level(v)$ is the maximum number of red vertices on any path ending at $v$, so this solution is feasible if and only if there does not exist a path in $G$ containing $L+1$ red vertices.

In terms of the problem we have been discussing, the DAG is a binary or an arithmetic circuit, the white vertices are the input variables, the blue vertices are the $\XOR$  (or addition) gates, the red vertices are the $\AND$ (or multiplication) gates,  $S$ is the set of ciphertexts that are bootstrapped during the computation, $\level(\cdot)$ is the noise level, and $L$ is the maximum allowed noise level.\footnote{Without loss of generality, we assume there are no \NOT gates, since they do not influence the noise level.}

 A feasible solution $S'\subseteq V$ is an \emph{$\alpha$-approximate solution} (for $\alpha\geq 1$) if $|S'|\leq \alpha\cdot \OPT$, where $\OPT$ denotes the minimum cardinality of a feasible solution.

\subsection{Results}
\label{sec:results}
We characterize the complexity of the bootstrap problem by providing a polynomial-time $L$-approximation algorithm (\cref{thm:approx}) and showing that, assuming the Unique Games Conjecture, $L$~is the best achievable approximation factor (\cref{thm:hardness}).

\begin{theorem}[approximation algorithm]
\label{thm:approx}
Let $L\geq 1$ be an integer parameter.
There is a deterministic polynomial-time approximation algorithm for the bootstrap problem within approximation factor $L$.
\end{theorem}
The proof of \cref{thm:approx} is in \cref{sec:approx}.

\begin{theorem}[hardness of approximation]
\label{thm:hardness}
Let $L\geq 2$ be an integer parameter.
For any $\epsilon>0$, it is NP-hard to approximate the bootstrap problem within a factor of $L-\epsilon$, assuming the Unique Games Conjecture.
\end{theorem}
The proof of \cref{thm:hardness} is in \cref{sec:hardness}.

\medskip

To design  the approximation algorithm used in Theorem~\ref{thm:hardness},  we first observe that a set of marked vertices is a feasible solution if and only if, for every path $p=v_1 \dots v_k$ that starts and ends at red vertices  and that traverses $L+1$ red vertices (including endpoints), at least one vertex among $v_1,\dots,v_{k-1}$ is marked.
Such path $p$ is called an \emph{interesting path}. 

Based on this observation, our algorithm starts by solving a linear program relaxation with one constraint for each interesting path and obtains a value $x_v\in [0,1]$, for every $v\in V$, indicating whether the vertex $v$ should be bootstrapped.
The challenging part of the algorithm is the rounding.

In a naive attempt to do the rounding, we define $\delta(u,v)$ as the $u$-to-$v$ distance in the metric induced by $\{x_v\}$.
In order that every interesting path from $u$ to $v$ contains a marked vertex, we choose a value $t\in[0,\delta(u,v)]$ (randomly or according to some rules such as in the \emph{region growing technique}~\cite{purplebook}), and then mark a vertex $w\in V$ if and only if $\delta(u,w)\leq t\leq \delta(u,w)+x_w$.
This approach does not yield a good approximation because $\delta(u,v)$ might be very small or even zero (see \cref{fig:counter-example-dist-rounding}), as there might exist short non-interesting $u$-to-$v$ paths.
Hence the major difference between the bootstrap problem and classical cut problems (e.g., min-cut, multi-cut, multi-terminal cut): in the bootstrap problem, for each pair of vertices $(u,v)$, we only want to ``cut'' the interesting $u$-to-$v$ paths, but the non-interesting $u$-to-$v$ paths may remain.

\begin{figure}[htb]
  \centering
  \begin{minipage}[c]{.4\textwidth}
    \centering
    \includegraphics{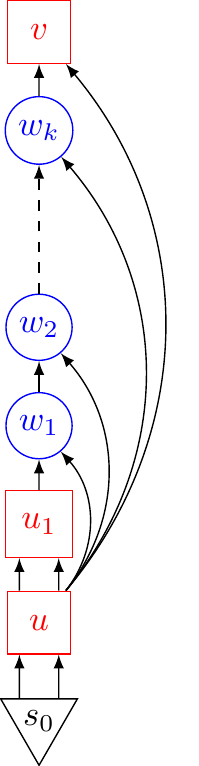}
  \end{minipage}
  \begin{minipage}[c]{.5\textwidth}
    \small
    Graph with $4 + k$ vertices. $L = 2$.\\
    The only interesting path is: $u u_1 w_1 w_2 \dots w_k v$.\\
    Circle vertices are blue, square vertices are red, and the triangle vertex is white.
    
    \medskip
    Suppose that the fractional solution obtained from the linear program is the following:
    \[
      \begin{cases}
        x_{w_i} = \frac{1}{k} &\text{for } i=1,\dots,k; \\
        x_{u} = x_{u_1} = x_{v} = 0.
      \end{cases}
    \]
    Then $\delta(u,v)=0$.

    
  \end{minipage}

  \caption{Example of circuit for which naive rounding does not work}
  \label{fig:counter-example-dist-rounding}
\end{figure}

Another attempt is to apply \emph{iterative rounding}\cite{jain2001factor,purplebook}.
However this does not seem to help for the bootstrap problem, mainly because the family of interesting paths is not closed under union, intersection, and difference.

Our approach to rounding is instead to separate paths according to the number of red vertices which they traverse and define a  {function} $f_i$ with respect to all paths that traverses exactly $i$ red vertices.
This is a main idea of the algorithm.
We perform a rounding for each function separately and obtain $L$ sets of marked vertices.
By taking the union of these sets, we obtain a feasible solution of cardinality at most $L\cdot \OPT$.

\begin{remark}
When $L=1$, the output of algorithm in \cref{thm:approx} is optimal.
\end{remark}

To prove the lower bound of Theorem~\ref{thm:hardness}, we first look at a related problem called the \emph{DAG Vertex Deletion (DVD)} problem~\cite{DBLP:journals/tc/PaikRS94,svensson2013hardness}.
In the DVD problem, we are given a directed acyclic graph~$H$ and an integer $L\geq 2$ and we want to delete the minimum number of vertices so that the resulting graph has no path containing $L$ vertices.
Svensson~\cite{svensson2013hardness} showed that approximating DVD within an $L-\epsilon$ factor is NP-hard, assuming the Unique Games Conjecture.

To show the UG-hardness of approximating the bootstrap problem, we provide an approximation-preserving reduction from the DVD problem to the bootstrap problem.

\subsection{Discussion of the Model}

\paragraph{Necessity of bootstrappings.} To date, the bootstrapping paradigm
is the only known way of obtaining an unbounded FHE scheme, i.e., one that can
homomorphically evaluate any efficient function using constant-size keys and ciphertexts. Therefore, to exploit the full potential
of fully homomorphic encryption, one must resort to bootstrapping.

\paragraph{Noise levels.} The bootstrap problem is a simplification of  the way in which FHE
schemes behave, since in practice noise
grows in a more complex manner. Indeed, all encryption procedures in fully
homomorphic encryption schemes consist of adding a short noise to the encoding
of the message (a bit or more generally an integer). Since the noise is {added} to the encoding, and computing
an \XOR gate homomorphically essentially corresponds to adding the
ciphertexts and thus adding the corresponding noises, on a logarithmic scale the amount of noise remains approximately as large as the maximum input noise
(up to one bit). On the other hand, computing an \AND gate requires a multiplication of the ciphertexts, and makes the noise growth noticeably larger~\cite{C:HalSho14}. This
is why the cryptographic community introduced the simplified model of~\cref{sec:motivation} and started building circuits for which the noise does not increase too much in this model~\cite{fse:CanteautCFLNPS16,EC:ARSTZ15}.

In this paper, we say that the noise level of a ciphertext is in
$[0, L]$, where $L$ is the parameter of the FHE scheme. In previous
works, the level was either in $[1, \ell_{\max}]$~\cite{FCW:LepPai13,SAC:PaiVia15}, or in $[9,
\ell_{\max}]$~\cite{C:HalSho14},\footnote{The $9$ comes from the fact that a
``fresh'' ciphertext (i.e., an unprocessed encryption of a bit) is said to
have noise level $1$, and that after a bootstrapping, the resulting ciphertext
has a noise level $9$.} where $\ell_{\max} =
L+1$. 
This is equivalent: $\ell=0$ should not be interpreted to mean that a ciphertext is
noise-free, but that the amount of noise the ciphertext contains results from a
bootstrapping operation.

\paragraph{Other noise behaviors.}
The noise model in \cref{sec:motivation} corresponds to the family of FHE schemes that are the most efficient in practice, but there exist other families of FHE schemes.

One family corresponds to the first implementations that were
proposed~\cite{EC:GenHal11,EC:CorNacTib12,EC:CCKLLT13}.
Therein, non-linear gates behave as $\cdot +\cdot$ for the
noise levels.\footnote{For this model, the lowest noise level needs to be set to~$1$ instead of~$0$, and the maximum allowed noise is $L+1$ instead of~$L$.}
Hence, the noise growth is exponential with the multiplicative
depth of the circuit and these schemes will never be used in practice.
In addition, to get reasonable parameters, the proof-of-concept implementations set $L=1$ (compared to, e.g., $L=41$ in the HElib
implementation~\cite{EC:HalSho15}) and in this particular case, the noise model is actually equivalent to the one we are considering (when $L=1$).

Another
family is the one of the GSW scheme~\cite{C:GenSahWat13,ITCS:BraVai14}. A variant of this scheme has been implemented by Ducas and Micciancio~\cite{EC:DucMic15}. The latter implementation has a faster wall-clock time for bootstrapping than~\cite{EC:CCKLLT13,EC:HalSho15}, but does not support large plaintext spaces nor vector plaintexts, hence it has larger amortized per-bit timing. The noise
behavior is slightly different there: it is asymmetric (i.e., the order of the
inputs matters). Multiplying a ciphertext $c_i$ by a ciphertext $c_j$, of
respective noise levels $\ell_i$ and $\ell_j$, yields a ciphertext of noise
level $\ell_i+1$.

\paragraph{Computing model.}
In the bootstrap problem, we minimize the total number of bootstrappings (i.e., marked vertices), thus accounting for classical sequential complexity.
We could also consider a parallel computing model, where doing any number of bootstrappings in parallel cost the same as doing one bootstrapping.
This might be relevant in some Cloud-based scenarios where the user encrypting the data has an unbounded amount of money and only want to minimize the time to get the result of the circuit evaluation over the encrypted data.
However, the financial cost would basically be proportional to the total number of bootstrappings.
Furthermore, we remark that this parallel version of the bootstrap problem has a trivial solution: topologically sorting the DAG and greedily marking the vertices with noise level greater than~$L$.

\fabr{Ok? Say more? In particular say that this solution might $\Omega(n)$ more vertices than required? If yes, add a Figure? Propose open problems with a fix number of parallel processes?}\tar{OK for me - I think we don't need to say more; the things you are proposing could go in a conclusion if we had one.}

\subsection{Other Related Work}
\label{sec:intro-dvd}

The DVD problem (see \cref{sec:results}) was introduced by Paik, Reddy, and Sahni in~\cite{DBLP:journals/tc/PaikRS94} in the context of certain VLSI design and communication problems.
Svensson showed that the DVD problem is UG-hard~\cite{svensson2013hardness}. 
His work was mainly motivated by the classical \emph{Discrete Time–Cost Tradeoff Problem} in the completely different setting of \emph{Project Scheduling}.


\section{Approximation Algorithm}
\label{sec:approx}

To prove \cref{thm:approx}, we give a randomized algorithm (\cref{alg:approx}) in \cref{sec:algo}, analyze it in \cref{sec:analysis}, and  derandomize it  in \cref{sec:derandomize}.

\subsection{Algorithm}
\label{sec:algo}


For a path $p=v_1\dots v_k$, the vertex $v_k$ is called the \emph{final vertex} of $p$ and the vertices $v_1,\dots,v_{k-1}$ are called the \emph{non-final vertices} of $p$.

The following fact is used throughout the paper. Its proof is in \cref{sec:proof-equi}.
\begin{fact}
\label{fact:equi}
A set of marked vertices is a feasible solution if and only if every path that starts and ends at a red vertex and that contains exactly  $L+1$ red vertices  (including endpoints) has  a non-final vertex that is marked.
\end{fact}

This fact leads to the definition of \emph{interesting} paths.

\begin{definition}[interesting path]
A path in $G$ is called \emph{interesting} if it starts and ends at red vertices, and traverses exactly $L+1$ red vertices  (including endpoints). For a given vertex $v\in V$ and a given level $i\in \{ 1,\ldots , L+1\}$, a path in $G$ is called \emph{$(v,i)$-interesting} if it starts at a red vertex, ends at $v$, and traverses exactly $i$ red vertices (including endpoints, if appropriate).
\end{definition}

We associate to each vertex $v\in V$ a non-negative weight $x_v$.
In our algorithm, these weights come from a solution of a linear program (LP).
These weights induce a metric.
More formally, we define the following notion of length.

\begin{definition}[length]
Let $p$ be a path in $G$.
We define the \emph{length} $\len(p)$ of $p$ as the sum of the weights~$x_v$ of all the non-final vertices $v$ of $p$.
We define $f_i(v)$ as the minimum length of a $(v,i)$-interesting path.\footnote{$f_i(v):=\infty$ if there is no $(v,i)$-interesting path.}
\end{definition}

We remark that, for any red vertex $v\in V$, a $(v,L+1)$-interesting path is an interesting path.

\newcommand{\dist}{\delta}

\begin{algorithm}[t]
\caption{Approximation algorithm for the bootstrap problem}
\label{alg:approx}
\begin{algorithmic}[1]
\State Solve the following LP relaxation, where we have one variable $x_v$ for each vertex $v$, representing whether there is a bootstrapping on $v$; and one constraint for each interesting path $p$.\label{alg-line:LP}
 \begin{align*}
\min \quad &\sum_{v\in V}x_v\\
\mathrm{s.\, t.}\quad&\sum_{\substack{\text{non-final}\\ \text{vertex $v$ of $p$}}} x_v \geq 1 & \text{$\forall$ interesting path $p$}\\
&0\le x_v \le 1 & \forall v\in V
\end{align*}
\State For every red vertex $u$ and blue vertex $v$, compute 
\[\delta(u,v)=\min \big\{ x_{u}+x_{v_2}+\cdots +x_{v_{k-1} }:  \text{ path } p=uv_2\ldots v_{k-1}v \text{ such that }v_2,\dots,v_{k-1} \text{ are blue}\big\},\]
using a classical shortest path algorithm.
By convention $\delta(u,v):=\infty$ if no such path exists.
\State For every vertex $v$ and integer $i\in \{1,\dots,L+1\}$, compute 
\[f_i(v)=\min \big\{ x_{v_1}+x_{v_2}+\cdots +x_{v_{k-1} }: \text{ path }p=v_1v_2\ldots v_{k-1}v \text{ is }(v,i)\text{-interesting}\big\}\]
using the side table $\delta$ and a dynamic program (see \cref{sec:algo}). 
By convention $f_i(v):=\infty$ if no such path exists.
\label{alg-line:DP}
\State Rounding: Pick a uniformly random value $t\in [0,1]$; A vertex $v$ is marked if and only if there exists $i\in \{1,\dots,L\}$ s.t.\ $t\in [f_i(v),f_i(v)+x_v]$. \label{alg-line:rounding}
\end{algorithmic}
\end{algorithm}

To compute ${\{f_i(v)\}}_{v,i}$, we use the following dynamic program, which proceeds in phases corresponding to $i=1,\dots,L+1$.
\begin{itemize}
\item For the base case $i=1$: 
$f_1(v)=
\begin{cases}
    \infty & \text{if $v$ is white},\\
	0 & \text{if $v$ is red},\\
	\displaystyle \min_{\text{$u$ red}} \delta(u,v) & \text{if $v$ is blue}.
\end{cases}$
\item For $i\in\{2,\dots,L+1\}$: 
$f_i(v)=
\begin{cases} 
    \infty & \text{if $v$ is white},\\
    \displaystyle \min_{(u,v) \in E} (f_{i-1}(u)+x_u)  & \text{if $v$ is red,}\\
	\displaystyle \min_{\text{$u$ red}} (f_{i}(u)+\delta(u,v)) & \text{if $v$ is blue}.
\end{cases}
$
\end{itemize}

\subsection{Analysis}
\label{sec:analysis}

We now prove that the output of \cref{alg:approx} is a feasible solution (correctness property) and has cardinality at most $L \cdot \OPT$ (approximation factor~$L$).
We then show that \cref{alg:approx} runs in polynomial time.

\paragraph{Analysis of Correctness.}
Consider an interesting path $p=v_1\dots v_k$.

\begin{lemma}
\label{lem:correct}
Let $p=v_1 \dots v_k$ be an interesting path.
For every $j\in \{1,\dots,k\}$, let $i_j\in \N$ denote the number of red vertices on the subpath $v_1 \dots v_j$ of $p$.
Then, for any $t\in[0,1]$, there exists $j\in \{1,\dots,k-1\}$ such that $t\in [f(v_j,i_j),f(v_j,i_j)+x_{v_j}]$.
\end{lemma}

Applying Lemma~\ref{lem:correct}, and using the fact that $i_j\in\{1,\dots,L\}$ for every $j\in \{1,\dots,k-1\}$, we see that the algorithm marks at least one non-final vertex of $p$, and so by  \cref{fact:equi} the output is a feasible solution, proving correctness.

\begin{proof} (Proof of Lemma~\ref{lem:correct})
By definition of interesting paths, the sequence ${\{i_j\}}_j$ is non-decreasing, $i_1=1$, $i_{k-1}=L$, and $i_k=L+1$.
It is sufficient to show that the interval $[0,1]$ is contained in the union of the intervals $[f(v_j,i_j),f(v_j,i_j)+x_{v_j}]$ over all $j\in \{1,\dots,k-1\}$, which is a direct consequence of the three following properties (see \cref{fig:intervals}):
\begin{enumerate}
\item $f(v_1,i_1)=0$;
\item for every $j\in\{1,\dots,k-1\}$, $f(v_{j+1},i_{j+1})\leq f(v_j,i_j)+x_{v_j}$;
\item $f(v_{k},i_{k})\geq 1$.
\end{enumerate}

The first property follows directly from the definition of $f$ since $v_1$ is red and $i_1=1$.

To show the second property, for any $j\in \{1,\dots,k-1\}$, consider a $(v_j,i_j)$-interesting path~$p'$ that achieves the length $f(v_j,i_j)$.
We observe that the concatenation of $p'$ and $v_{j+1}$ is a $(v_{j+1},i_{j+1})$-interesting path and it has length $f(v_j,i_j)+x_{v_j}$.
From the definition of $f(v_{j+1}, i_{j+1})$, we have $f(v_{j+1},i_{j+1})\leq f(v_j,i_j)+x_{v_j}$.

To show the third property, consider a $(v_k,i_k)$-interesting path $p'$ that achieves the length $f(v_k,i_k)$.
Then $p'$ is an interesting path since $v_k$ is red and $i_k=L+1$.
 ($p'$ may differ from $p$ though.)
Therefore, the constraint on $p'$ in the LP implies that $\len(p')\geq 1$.
Hence $f(v_k,i_k)=\len(p')\geq 1$.

This concludes the proof.
\end{proof}

\begin{figure}[tbh]
  \centering
  \includegraphics[width=\textwidth]{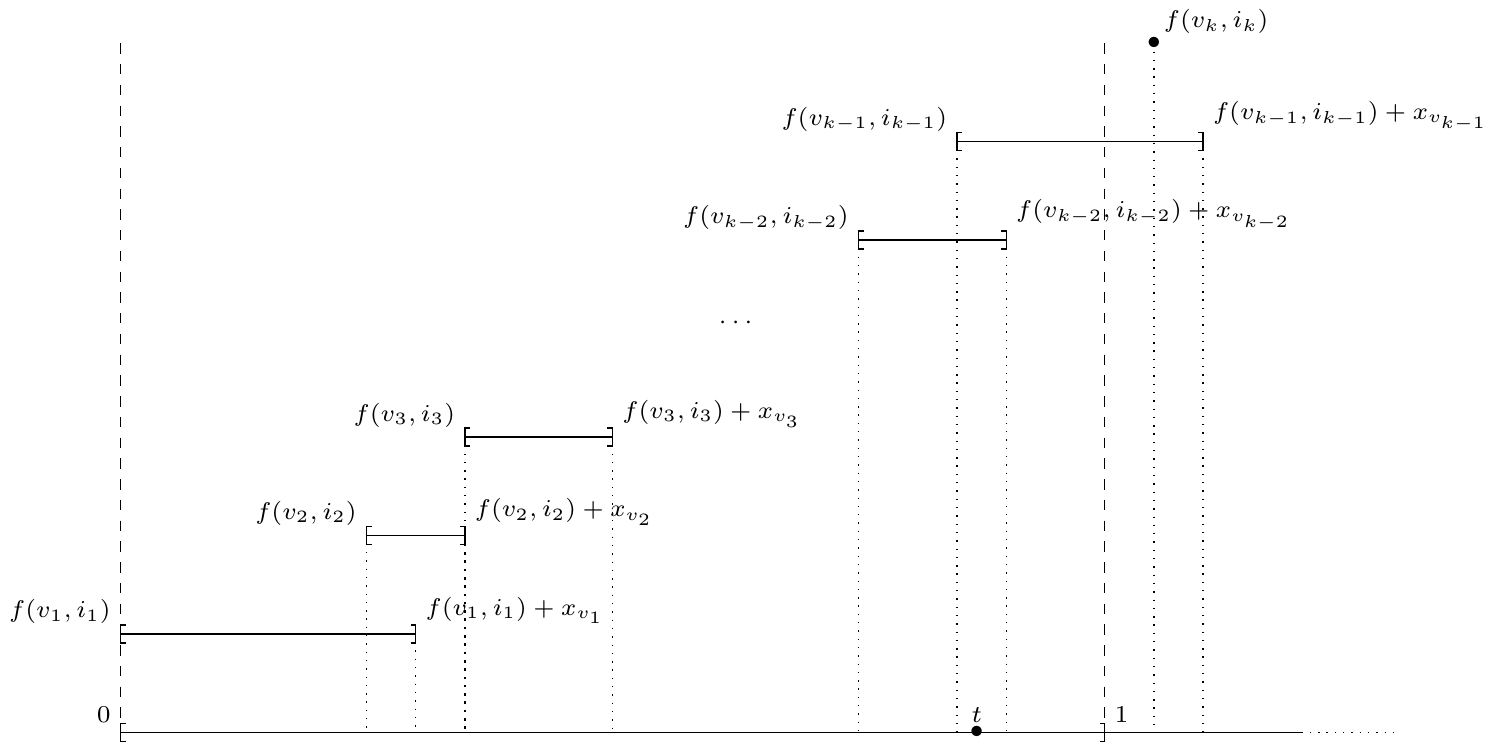}
  \caption{Illustration for the proof of \cref{lem:correct}}
  \label{fig:intervals}
\end{figure}

\paragraph{Analysis of Quality of Approximation.}
The expected value of the output is the expected number of marked vertices, $\sum_{v\in V} \Pr (v\text{ marked})$. Let $v\in V$.
By the algorithm and a union bound:
$$\Pr (v\text{ marked})=
\Pr (\exists i\in \{1,\dots,L\} ~:~t\in [f_i(v),f_i(v)+x_v])\leq 
\sum_i \Pr (t\in [f_i(v),f_i(v)+x_v]).$$
For $t$ uniformly random in $[0,1]$, the probability that $t\in [f_i(v),f_i(v)+x_v]$ is at most $x_v$. 
Thus $\Pr (v\text{ marked})\leq Lx_v$ and the expected value of the output is at most 
$L(\sum_{v\in V} x_v)$. Since the linear program is a relaxation of the problem, $\sum_v x_v$ is less than or equal to the optimum value  of the bootstrap problem, proving that the output is an $L$-approximation.

\paragraph{Analysis of Running Time.}
Clearly, computing ${\{f_i(v)\}}_{v,i}$ takes polynomial time.
Next we show that the LP in Step~\ref{alg-line:LP} of the algorithm can be solved in polynomial time (regardless of an exponential number of constraints).
To that end, it is well known (see, e.g.,~\cite{purplebook}) that a polynomial-time \emph{separation oracle}\footnote{A separation oracle takes as input a supposedly feasible solution to the linear program, and either verifies that it is indeed a feasible solution to the linear program or, if it is infeasible, produces a violated constraint.} for this LP suffices.

To check whether an oracle input ${\{x_v\}}_v$ is a feasible solution to the LP, we compute ${\{f_i(v)\}}_{v,i}$ with respect to ${\{x_v\}}_v$ using the same dynamic program as before, in polynomial time.
From the definition of $f$, we have that ${\{x_v\}}_v$ is a feasible solution if and only if $f_{L+1}(v)\geq 1$ for every red vertex $v\in V$.
Suppose there is some red vertex $v\in V$ with $f_{L+1}(v)<1$.
Then there must be an interesting path $p$ with final vertex $v$ such that the constraint on $p$ is violated.
It is easy to enrich the dynamic program in a standard manner, so that we obtain the entire path $p$.
Thus we complete the description of the polynomial-time separation oracle.

Therefore, the overall running time of \cref{alg:approx} is polynomial.

\subsection{Derandomization}
\label{sec:derandomize}
\cref{alg:approx} can be easily derandomized: ${\{f_i(v)\}}_{v,i}\cup {\{f_i(v)+x_v\}}_{v,i}$ contains at most $2 \lvert V\rvert \cdot L$ different values, so they separate the $[0,1]$ interval into at most $2 \lvert V\rvert \cdot L+1$ sub-intervals.
We can enumerate one value~$t$ for each sub-interval, compute a feasible solution with respect to each value $t$, and finally return the best solution among them.


\section{Hardness of Approximation}
\label{sec:hardness}

In this section, we prove \cref{thm:hardness}.
First, we recall the definition of the DAG Vertex Deletion (DVD) problem (see \cref{sec:results}).
In \cref{lem:reduction}, we reduce the DVD problem to the bootstrap problem.
The hardness of the bootstrap problem then follows from the hardness of the DVD problem~\cite{svensson2013hardness}. 

\begin{lemma}[Adapted from Theorem~1.1 in~\cite{svensson2013hardness}]
\label{lem:hardness-DVD}
Let $L\geq 2$ be an integer parameter.
For any $\epsilon>0$, it is NP-hard to approximate the DVD problem within a factor of $L-\epsilon$, assuming the Unique Games Conjecture.
\end{lemma}

\begin{lemma}
\label{lem:reduction}
There is an approximation-preserving reduction from the DVD problem to the bootstrap problem.
\end{lemma}

\cref{thm:hardness} follows immediately from \cref{lem:hardness-DVD,lem:reduction}.
In the rest of the section, we prove \cref{lem:reduction}.
The proof is elementary but delicate, mainly because in the bootstrap problem, vertices have indegree at most 2, while in the DVD problem, vertices may have arbitrary indegree.

\bigskip

Consider a DVD instance with the DAG $H=(V_H,E_H)$ and the integer parameter $L\geq 2$.
As a warm-up, let us first suppose that all vertices in $V_H$ have indegree at most $2$.
We construct a DAG $G=(V,E)$ for the bootstrap problem:
\begin{itemize}
\item We create a new vertex set $V_H'$ (which can be viewed as a clone of $V_H$): for every $v\in V_H$, $V_H'$ contains a new vertex $v'$.
We also create a new vertex $s_0$. 
The vertex set $V$ of $G$ is defined as $V:=V_H\cup V_H'\cup \{s_0\}$. 
The vertices in $V_H\cup V_H'$ are red, and the vertex $s_0$ is white.
\item The edge set $E$ of $G$ consists of all edges of $E_H$ and the following edges: for every vertex $v\in V_H$, 2 copies of the edge $(v,v')$ and ($2-\text{indegree}(v))$ copies of the edge $(s_0,v)$.
\end{itemize}

The following lemma implies that there is an approximation-preserving reduction.

\begin{lemma}
\label{lem:dvd-bootstrap}
A feasible solution to the DVD problem for the instance $(L,H)$ can be transformed into a feasible bootstrap solution for the instance $(L,G)$ with at most the same cardinality, and vice versa.
\end{lemma}
\begin{proof}
Let $S$ be a feasible solution to the DVD problem, i.e., every path of $L$ vertices in $H$ contains a vertex in $S$. 
We show that $S$ is a feasible bootstrap solution for the instance $(L,G)$.
Using \cref{fact:equi}, we only need to show that every interesting path in $G$ contains a non-final vertex that is in $S$.
Let $p=v_1 \dots v_k$ be an interesting path in $G$.
We observe that $v_1,\dots,v_k$ are all red vertices: $v_1$ is red by the definition of an interesting path, and $v_j$ (for any $2\leq j\leq k$) is red since it has positive indegree (and thus cannot be $s_0$). 
By the definition of an interesting path, we know that $p$ contains $L+1$ red vertices, so $k=L+1$.
We further observe that every $v_i$ (for any $1\leq i\leq k-1$) is in $V_H$ since $v_i$ has positive outdegree (and thus cannot be in $V_H'$).
Thus $v_1,\dots,v_{k-1}$ form a path of $L$ vertices in~$H$.
Since $S$ is a feasible solution to the DVD problem, at least one vertex among $v_1,\dots,v_{k-1}$ is in $S$.
Thus $p$ contains a non-final vertex that is in $S$.

Conversely, let $S$ be a feasible bootstrap solution.
We show that $S\cap V_H$ is a feasible solution to the DVD problem.
Let $p=v_1,\dots,v_L$ be a path of $L$ vertices in $H$.
We only need to show that $p$ contains a vertex in $S\cap V_H$.
We construct a path $p'$ in $G$ that is the concatenation of $p$ and the vertex~$v_L'$.
We remark that $p'$ starts and ends on red vertices and contains exactly $L+1$ red vertices.
Therefore, it is an interesting path.
From \cref{fact:equi}, $S$ contains a vertex $u$ that is a non-final vertex of~$p'$, i.e., $u\in\{v_1,\dots,v_L\}$, hence $u$ is on $p$.
Since $\{v_1,\dots,v_L\}\subseteq V_H$, $u\in S\cap V_H$.
Thus $p$ contains a vertex in $S\cap V_H$.
\end{proof}

Let us now deal with the general case where the vertices of $H$ have arbitrary indegrees.
First, we initialize the DAG $G$ using the same transformation as before, except that $G$ now contains $(2-\text{indegree}(v))$ copies of an edge $(s_0,v)$ only if the vertex $v$ has indegree at most~2 (rather than for any vertex $v$).
After this transformation, every red vertex in $G$ has indegree at least~$2$.
To transform~$G$ into a DAG for the bootstrap problem, we just need to deal with the red vertices with indegree at least $3$. 
Let $v$ be such a vertex. 
We observe that $v\in V_H$.
Let $v_1,\dots,v_d$ be the direct successors of $v$ in $H$.
We remove from~$G$ the edges $(v_i,v)$ (for each $i$) and add to $G$ new blue vertices $w_1^{(v)},\dots,w_d^{(v)}$ and the following edges:
\begin{enumerate}[noitemsep,nolistsep]
\item two copies of the edge $(v_1,w_1^{(v)})$,
\item for $i=2,\dots,d$, an edge $(w_{i-1}^{(v)},w_i^{(v)})$ and an edge $(v_i,w_i^{(v)})$,
\item two copies of the edge $(w_d^{(v)},v)$.
\end{enumerate}
The transformation is depicted in \cref{fig:reduceindegree}.
Let $G=(V,E)$ be the final graph.
We can verify that $(L,G)$ is an instance of the bootstrap problem.

We show that \cref{lem:dvd-bootstrap} holds in this general setting.
The transformation from a feasible DVD solution to a feasible bootstrap solution is a trivial extension from the previous proof.
Let us now focus on the transformation from a feasible bootstrap solution to a feasible DVD solution.
The following proposition is the key to the proof.
\begin{proposition}
\label{prop:move-to-red}
Let $S$ be a feasible bootstrap solution for $(L,G)$ which contains a blue vertex $w_i^{(v)}$ for some $v\in V_H$ and some integer $i$.
Then $(S\setminus \{w_i^{(v)}\})\cup \{v\}$ is also a feasible bootstrap solution for $(L,G)$.
\end{proposition}

\begin{proof}
Let $S'=(S\setminus \{w_i^{(v)}\})\cup \{v\}$.
From \cref{fact:equi}, we only need to prove that every interesting path contains a non-final vertex that is in $S'$.
Since $S$ is a feasible bootstrap solution, the only non-trivial part is to prove that every interesting path ending in $v$ contains a non-final vertex that is in $S'$.
Let $p=v_1 \dots v_k$ be such a path, and let $j<k$ be the index of the last non-final red vertex of~$p$ (i.e., $v_j$ is a red vertex and there is no red vertex among $v_{j+1},\dots,v_{k-1}$).
Let $p'$ be the path $v_1 \dots v_j v_j'$.
Since there are exactly $L$ red vertices among $v_1,\dots,v_L$ and $v_j'$ is red, $p'$ is an interesting path. 
Thus some non-final vertex $u$ of $p'$ (i.e., $u\in\{v_1,\dots,v_j\}$) is in $S$.
We observe that $p'$ cannot contain $w_i^{(v)}$ since $w_i^{(v)}\in \{v_{j+1},\dots,v_{k-1}\}$.
Thus $u\in S'$ and is a non-final vertex of $p$.
\end{proof}

Let $S$ be a feasible bootstrap solution.
We construct another feasible bootstrap solution $S'$ with $|S'|\leq |S|$ such that $S'$ only contains red and white vertices.
As soon as $S$ contains a blue vertex, let it be~$w_i^{(v)}$ for some $v\in V_H$ and some integer $i$, we replace the blue vertex $w_i^{(v)}$ in $S$ by the red vertex~$v$. 
Let $S'$ be the final $S$.
Then $S'$ contains only red and white vertices and has cardinality at most $|S|$.
From \cref{prop:move-to-red}, $S'$ is a feasible bootstrap solution.

We now show that $S'\cap V_H$ is a solution to the DVD problem using similar arguments as before.
Let $p=v_1 \dots v_L$ be a path of $L$ vertices in $H$.
We only need to show that $p$ contains a vertex in $S'\cap V_H$.
We construct a (unique) path $p'$ in $G$, which starts at $v_1$, goes through $v_2,\dots,v_L$, and ends at the red vertex $v'_L$.
We remark that $p'$ starts and ends on red vertices and contains exactly $L+1$ red vertices, namely $\{v_1,\dots,v_L,v_L'\}$.
Therefore, it is an interesting path.
From \cref{fact:equi}, $S'$ contains a vertex $u$ that is a non-final vertex of $p'$.
Since $S'$ contains only red and white vertices, we have $u\in\{v_1,\dots,v_L\}$, hence $u$ is on $p$.
Since $\{v_1,\dots,v_L\}\subseteq V_H$, $u\in S'\cap V_H$.
Thus $p$ contains a vertex in $S'\cap V_H$.

This concludes the proof of \cref{lem:reduction}.

  \begin{figure}[htb]
    \centering
    \begin{tikzpicture}
      \tikzset{arrowstyle/.style={draw=black,single arrow,minimum height=#1, single arrow,
    single arrow head extend=.2cm,}}
      \node (a) at (0,0) {\includegraphics[scale=1]{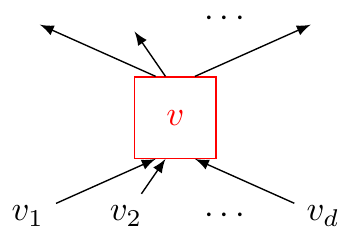}};
      \node[right=2.5cm of a] (b) {\includegraphics[scale=1]{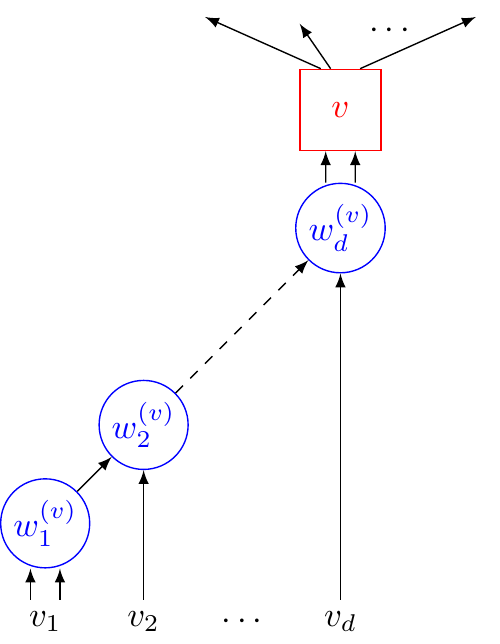}};
      \path (a) -- (b) node[midway,arrowstyle=1cm] {\vspace{0.1cm}};
    \end{tikzpicture}
    \caption{Reducing of the Indegree of a Vertex $v$ (in the Proof of \cref{lem:reduction}). \small Circle vertices are blue, square vertices are red, $v_{1},\dots,v_d$ are the direct predecessors of $v$, $w_1^{(v)},\dots,w_d^{(v)}$ are new blue vertices.}
    \label{fig:reduceindegree}
  \end{figure}


\paragraph{Acknowledgements.}
We would like to thank Florian Bourse and Pierrick Méaux for their discussion on other FHE models.

\newpage
\bibliographystyle{plain}
\bibliography{cryptobib/abbrev3,cryptobib/crypto,add}

\begin{thebibliography}{10}

\bibitem{EC:ARSTZ15}
Martin~R. Albrecht, Christian Rechberger, Thomas Schneider, Tyge Tiessen, and
  Michael Zohner.
\newblock Ciphers for {MPC} and {FHE}.
\newblock In Elisabeth Oswald and Marc Fischlin, editors, {\em EUROCRYPT~2015,
  Part I}, volume 9056 of {\em {LNCS}}, pages 430--454. Springer, Heidelberg,
  April 2015.

\bibitem{C:AlpPei13}
Jacob {Alperin-Sheriff} and Chris Peikert.
\newblock Practical bootstrapping in quasilinear time.
\newblock In Ran Canetti and Juan~A. Garay, editors, {\em CRYPTO~2013, Part I},
  volume 8042 of {\em {LNCS}}, pages 1--20. Springer, Heidelberg, August 2013.

\bibitem{DBLP:journals/jbi/BosLN14}
Joppe~W. Bos, Kristin~E. Lauter, and Michael Naehrig.
\newblock Private predictive analysis on encrypted medical data.
\newblock {\em Journal of Biomedical Informatics}, 50:234--243, 2014.

\bibitem{TOCT:BraGenVai14}
Zvika Brakerski, Craig Gentry, and Vinod Vaikuntanathan.
\newblock ({L}eveled) fully homomorphic encryption without bootstrapping.
\newblock {\em {TOCT}}, 6(3):13, 2014.

\bibitem{SIAM:BraVai14}
Zvika Brakerski and Vinod Vaikuntanathan.
\newblock Efficient fully homomorphic encryption from (standard) {LWE}.
\newblock {\em {SIAM} J. Comput.}, 43(2):831--871, 2014.

\bibitem{ITCS:BraVai14}
Zvika Brakerski and Vinod Vaikuntanathan.
\newblock Lattice-based {FHE} as secure as {PKE}.
\newblock In Moni Naor, editor, {\em ITCS 2014}, pages 1--12. {ACM}, January
  2014.

\bibitem{fse:CanteautCFLNPS16}
Anne Canteaut, Sergiu Carpov, Caroline Fontaine, Tancr{\`{e}}de Lepoint,
  Mar{\'{\i}}a Naya{-}Plasencia, Pascal Paillier, and Renaud Sirdey.
\newblock Stream ciphers: A practical solution for efficient
  homomorphic-ciphertext compression.
\newblock In {\em FSE 2016}, 2016.

\bibitem{EC:CCKLLT13}
Jung~Hee Cheon, Jean-S{\'e}bastien Coron, Jinsu Kim, Moon~Sung Lee,
  Tancr{\`e}de Lepoint, Mehdi Tibouchi, and Aaram Yun.
\newblock Batch fully homomorphic encryption over the integers.
\newblock In Thomas Johansson and Phong~Q. Nguyen, editors, {\em
  EUROCRYPT~2013}, volume 7881 of {\em {LNCS}}, pages 315--335. Springer,
  Heidelberg, May 2013.

\bibitem{C:CMNT11}
Jean-S{\'e}bastien Coron, Avradip Mandal, David Naccache, and Mehdi Tibouchi.
\newblock Fully homomorphic encryption over the integers with shorter public
  keys.
\newblock In Phillip Rogaway, editor, {\em CRYPTO~2011}, volume 6841 of {\em
  {LNCS}}, pages 487--504. Springer, Heidelberg, August 2011.

\bibitem{EC:CorNacTib12}
Jean-S{\'e}bastien Coron, David Naccache, and Mehdi Tibouchi.
\newblock Public key compression and modulus switching for fully homomorphic
  encryption over the integers.
\newblock In David Pointcheval and Thomas Johansson, editors, {\em
  EUROCRYPT~2012}, volume 7237 of {\em {LNCS}}, pages 446--464. Springer,
  Heidelberg, April 2012.

\bibitem{EC:DucMic15}
L{\'e}o Ducas and Daniele Micciancio.
\newblock {FHEW}: Bootstrapping homomorphic encryption in less than a second.
\newblock In Elisabeth Oswald and Marc Fischlin, editors, {\em EUROCRYPT~2015,
  Part I}, volume 9056 of {\em {LNCS}}, pages 617--640. Springer, Heidelberg,
  April 2015.

\bibitem{STOC:Gentry09}
Craig Gentry.
\newblock Fully homomorphic encryption using ideal lattices.
\newblock In Michael Mitzenmacher, editor, {\em 41st ACM STOC}, pages 169--178.
  {ACM} Press, May~/~June 2009.

\bibitem{C:Gentry10}
Craig Gentry.
\newblock Toward basing fully homomorphic encryption on worst-case hardness.
\newblock In Tal Rabin, editor, {\em CRYPTO~2010}, volume 6223 of {\em {LNCS}},
  pages 116--137. Springer, Heidelberg, August 2010.

\bibitem{FOCS:GenHal11}
Craig Gentry and Shai Halevi.
\newblock Fully homomorphic encryption without squashing using depth-3
  arithmetic circuits.
\newblock In Rafail Ostrovsky, editor, {\em 52nd FOCS}, pages 107--109. {IEEE}
  Computer Society Press, October 2011.

\bibitem{EC:GenHal11}
Craig Gentry and Shai Halevi.
\newblock Implementing {Gentry}'s fully-homomorphic encryption scheme.
\newblock In Kenneth~G. Paterson, editor, {\em EUROCRYPT~2011}, volume 6632 of
  {\em {LNCS}}, pages 129--148. Springer, Heidelberg, May 2011.

\bibitem{C:GenHalSma12}
Craig Gentry, Shai Halevi, and Nigel~P. Smart.
\newblock Homomorphic evaluation of the {AES} circuit.
\newblock In Reihaneh Safavi-Naini and Ran Canetti, editors, {\em CRYPTO~2012},
  volume 7417 of {\em {LNCS}}, pages 850--867. Springer, Heidelberg, August
  2012.

\bibitem{C:GenSahWat13}
Craig Gentry, Amit Sahai, and Brent Waters.
\newblock Homomorphic encryption from learning with errors:
  Conceptually-simpler, asymptotically-faster, attribute-based.
\newblock In Ran Canetti and Juan~A. Garay, editors, {\em CRYPTO~2013, Part I},
  volume 8042 of {\em {LNCS}}, pages 75--92. Springer, Heidelberg, August 2013.

\bibitem{IANDC:GolHagMun08}
Judy Goldsmith, Matthias Hagen, and Martin Mundhenk.
\newblock Complexity of {DNF} minimization and isomorphism testing for monotone
  formulas.
\newblock {\em Inf. Comput.}, 206(6):760--775, 2008.

\bibitem{ICISC:GraLauNae12}
Thore Graepel, Kristin Lauter, and Michael Naehrig.
\newblock {ML} confidential: Machine learning on encrypted data.
\newblock In Taekyoung Kwon, Mun{-}Kyu Lee, and Daesung Kwon, editors, {\em
  ICISC 12}, volume 7839 of {\em {LNCS}}, pages 1--21. Springer, Heidelberg,
  November 2013.

\bibitem{C:HalSho14}
Shai Halevi and Victor Shoup.
\newblock Algorithms in {HElib}.
\newblock In Juan~A. Garay and Rosario Gennaro, editors, {\em CRYPTO~2014, Part
  I}, volume 8616 of {\em {LNCS}}, pages 554--571. Springer, Heidelberg, August
  2014.

\bibitem{EC:HalSho15}
Shai Halevi and Victor Shoup.
\newblock Bootstrapping for {HElib}.
\newblock In Elisabeth Oswald and Marc Fischlin, editors, {\em EUROCRYPT~2015,
  Part I}, volume 9056 of {\em {LNCS}}, pages 641--670. Springer, Heidelberg,
  April 2015.

\bibitem{jain2001factor}
Kamal Jain.
\newblock A factor 2 approximation algorithm for the generalized {Steiner}
  network problem.
\newblock {\em Combinatorica}, 21(1):39--60, 2001.

\bibitem{LC:LauLopNae14}
Kristin~E. Lauter, Adriana {L{\'o}pez-Alt}, and Michael Naehrig.
\newblock Private computation on encrypted genomic data.
\newblock In Diego~F. Aranha and Alfred Menezes, editors, {\em
  LATINCRYPT~2014}, volume 8895 of {\em {LNCS}}, pages 3--27. Springer,
  Heidelberg, September 2015.

\bibitem{FCW:LepPai13}
Tancr{\`e}de Lepoint and Pascal Paillier.
\newblock On the minimal number of bootstrappings in homomorphic circuits.
\newblock In Andrew~A. Adams, Michael Brenner, and Matthew Smith, editors, {\em
  FC 2013 Workshops}, {LNCS}, pages 189--200. Springer, Heidelberg, April 2013.

\bibitem{STOC:LopTroVai12}
Adriana {L{\'o}pez-Alt}, Eran Tromer, and Vinod Vaikuntanathan.
\newblock On-the-fly multiparty computation on the cloud via multikey fully
  homomorphic encryption.
\newblock In Howard~J. Karloff and Toniann Pitassi, editors, {\em 44th ACM
  STOC}, pages 1219--1234. {ACM} Press, May 2012.

\bibitem{DBLP:conf/ccs/NaehrigLV11}
Michael Naehrig, Kristin~E. Lauter, and Vinod Vaikuntanathan.
\newblock Can homomorphic encryption be practical?
\newblock In {\em {CCSW}}, pages 113--124. {ACM}, 2011.

\bibitem{EC:NuiKur15}
Koji Nuida and Kaoru Kurosawa.
\newblock ({B}atch) fully homomorphic encryption over integers for non-binary
  message spaces.
\newblock In Elisabeth Oswald and Marc Fischlin, editors, {\em EUROCRYPT~2015,
  Part I}, volume 9056 of {\em {LNCS}}, pages 537--555. Springer, Heidelberg,
  April 2015.

\bibitem{DBLP:journals/tc/PaikRS94}
Doowon Paik, Sudhakar~M. Reddy, and Sartaj Sahni.
\newblock Deleting vertices to bound path length.
\newblock {\em {IEEE} Trans. Computers}, 43(9):1091--1096, 1994.

\bibitem{SAC:PaiVia15}
Marie Paindavoine and Bastien Vialla.
\newblock Minimizing the number of bootstrappings in fully homomorphic
  encryption.
\newblock In {\em {SAC}}, volume 9566 of {\em Lecture Notes in Computer
  Science}, pages 25--43. Springer, 2015.

\bibitem{FCW:RohCou14}
Kurt Rohloff and David~Bruce Cousins.
\newblock A scalable implementation of fully homomorphic encryption built on
  {NTRU}.
\newblock In Rainer B{\"o}hme, Michael Brenner, Tyler Moore, and Matthew Smith,
  editors, {\em FC 2014 Workshops}, volume 8438 of {\em {LNCS}}, pages
  221--234. Springer, Heidelberg, March 2014.

\bibitem{svensson2013hardness}
Ola Svensson.
\newblock Hardness of vertex deletion and project scheduling.
\newblock {\em Theory of Computing}, 9(24):759--781, 2013.

\bibitem{EC:DGHV10}
Marten van Dijk, Craig Gentry, Shai Halevi, and Vinod Vaikuntanathan.
\newblock Fully homomorphic encryption over the integers.
\newblock In Henri Gilbert, editor, {\em EUROCRYPT~2010}, volume 6110 of {\em
  {LNCS}}, pages 24--43. Springer, Heidelberg, May 2010.

\bibitem{purplebook}
David~P. Williamson and David~B. Shmoys.
\newblock {\em The Design of Approximation Algorithms}.
\newblock Cambridge University Press, New York, NY, USA, 1st edition, 2011.

\end{thebibliography}

\newpage
\appendix

\section{Proof of Fact~\ref{fact:equi}}
\label{sec:proof-equi}

We first observe that a set $S$ of marked vertices is a feasible solution if and only if, for all \emph{red} vertices $u\in V$, the noise level $\ell(u)$ with respect to $S$ is at most $L$.
This is because a white vertex has noise level 0 and a blue vertex has noise level not exceeding those of its direct predecessors. 
We further observe that, for each red vertex $u\in V$,  the level $\ell(u)$ is the maximum number of red vertices on \emph{any} $v$-to-$u$ path (for some red vertex $v$) that does not contain any marked vertices (except $u$ if appropriate).
Thus $\ell(u)\leq L$ if and only if every path that starts at a red vertex and ends at $u$ and that contains exactly  $L+1$ red vertices  (including endpoints) has  a non-final vertex that is marked.
Therefore, all \emph{red} vertices $u\in V$ are such that $\ell(u)\leq L$ if and only if every path that starts and ends at a red vertex and that contains exactly  $L+1$ red vertices  (including endpoints) has  a non-final vertex that is marked. \qed


\end{document}